\documentclass[11pt,letterpaper,reqno]{amsart}

\title[The Siegel Upper Half Space is a Marsden--Weinstein Quotient]{The Siegel Upper Half Space is a Marsden--Weinstein Quotient: Symplectic Reduction and Gaussian Wave Packets}
\author{Tomoki Ohsawa}
\address{Department of Mathematical Sciences, The University of Texas at Dallas, 800 W Campbell Rd, Richardson, TX 75080-3021}
\email{tomoki@utdallas.edu}

\date{\today}



\usepackage{graphicx,amssymb,paralist}
\usepackage[margin=1in, marginpar=.5in]{geometry}

\usepackage[numbers,sort&compress]{natbib}
\usepackage[colorlinks=false]{hyperref}

\usepackage{tikz-cd}

\theoremstyle{plain}
\newtheorem{theorem}{Theorem}[section]

\newtheorem{proposition}[theorem]{Proposition}

\theoremstyle{definition}

\theoremstyle{remark}
\newtheorem{remark}[theorem]{Remark}

\def\od#1#2{\dfrac{d#1}{d#2}}
\def\pd#1#2{\dfrac{\partial #1}{\partial #2}}

\def\parentheses#1{{\left(#1\right)}}
\def\brackets#1{{\left[#1\right]}}
\def\braces#1{{\left\{#1\right\}}}

\def\tr{\mathop{\mathrm{tr}}\nolimits}

\def\pr{\mathop{\mathrm{pr}}\nolimits}

\def\norm#1{{\left\|#1\right\|}}

\def\DS{\displaystyle}
\def\R{\mathbb{R}}
\def\C{\mathbb{C}}

\def\defeq{\mathrel{\mathop:}=}
\def\eqdef{=\mathrel{\mathop:}}
\def\setdef#1#2{{\left\{ #1 \ |\ #2 \right\}}}
\def\ip#1#2{{\left\langle#1,#2\right\rangle}}

\newcommand{\id}{\operatorname{id}}
\renewcommand{\Re}{\operatorname{Re}}
\renewcommand{\Im}{\operatorname{Im}}
\def\eps{\varepsilon}
\def\Mat{\mathsf{M}}
\def\SO{\mathsf{SO}}

\def\Sp{\mathsf{Sp}}
\def\Orth{\mathsf{O}}
\def\U{\mathsf{U}}
\def\so{\mathfrak{so}}

\def\orth{\mathfrak{o}}

\newenvironment{tbmatrix}{\left[\begin{smallmatrix}}{\end{smallmatrix}\right]}

\def\d{{\bf d}}
\def\ins#1{{\bf i}_{#1}}

\newcommand\Ad{\operatorname{Ad}}

\begin{document}

\footskip=.6in

\begin{abstract}
  We show that the Siegel upper half space $\Sigma_{d}$ is identified with the Marsden--Weinstein quotient obtained by symplectic reduction of the cotangent bundle $T^{*}\mathbb{R}^{2d^{2}}$ with $\mathsf{O}(2d)$-symmetry.
  The reduced symplectic form on $\Sigma_{d}$ corresponding to the standard symplectic form on $T^{*}\mathbb{R}^{2d^{2}}$ turns out to be a constant multiple of the symplectic form on $\Sigma_{d}$ obtained by Siegel.
  Our motivation is to understand the geometry behind two different formulations of the Gaussian wave packet dynamics commonly used in semiclassical mechanics.
  Specifically, we show that the two formulations are related via the symplectic reduction.
\end{abstract}

\maketitle

\section{Introduction}
\subsection{The Siegel Upper Half Space}
The set $\Sigma_{d}$ of symmetric $d \times d$ complex matrices (symmetric in the real sense) with positive-definite imaginary parts, i.e.,
\begin{equation*}
  \Sigma_{d} \defeq 
  \setdef{ \mathcal{A} + {\rm i}\mathcal{B} \in \Mat_{d}(\mathbb{C}) }{ \mathcal{A}^{T} = \mathcal{A},\, \mathcal{B}^{T} = \mathcal{B},\, \mathcal{B} > 0 },
\end{equation*}
is called the {\em Siegel upper half space}.
For $d = 1$, one easily sees that $\Sigma_{1}$ is the upper half plane $\mathbb{H} \defeq \setdef{ a + {\rm i}b \in \C }{ b > 0 }$ of the complex plane $\C$, and hence one may think of $\Sigma_{d}$ as a generalization of the upper half plane $\mathbb{H}$ to higher dimensions.
It is well known that the upper half plane $\mathbb{H}$ may be identified with the homogeneous space ${\sf Sp}(2,\R)/\U(1)$ via the linear fractional or M\"obius transformation $ z \mapsto (a z + b)(c z + d)^{-1}$.
The Siegel upper half space $\Sigma_{d}$ is a natural generalization of $\mathbb{H}$ to higher dimensions in the sense that $\Sigma_{d}$ is identified with the homogeneous space $\Sp(2d,\R)/\U(d)$ via a generalized linear fractional transformation; see Section~\ref{sec:Sigma_d} for details.

\subsection{Motivation: Gaussian Wave Packet Dynamics}
Our motivation for studying the geometry of the Siegel upper half space $\Sigma_{d}$ is to better understand the underlying geometry for the dynamics of the Gaussian wave packet
\begin{equation}
  \label{eq:chi}
  \chi(x) = \exp\braces{ \frac{{\rm i}}{\hbar}\brackets{ \frac{1}{2}(x - q)^{T}(\mathcal{A} + {\rm i}\mathcal{B})(x - q) + p \cdot (x - q) + (\phi + {\rm i}\delta) } },
\end{equation}
which is widely used in the study of the semiclassical limit of the Schr\"odinger equation.
It is parametrized by $(q,p) \in T^{*}\R^{d}$, $\mathcal{A} + {\rm i}\mathcal{B} \in \Sigma_{d}$, $\phi \in \mathbb{S}^{1}$, and $\delta \in \R$, and it is well known (see \citet{He1975a,He1976b} and \citet{Ha1980, Ha1998}) that \eqref{eq:chi} is an {\em exact} solution of the time-dependent Schr\"odinger equation with quadratic potentials if these parameters, as functions of the time, satisfy a certain set of ODEs.

The geometry of $\Sigma_{d}$---particularly the fact that $\Sigma_{d}$ is a symplectic manifold---becomes important when one tries to understand the set of ODEs as a Hamiltonian system on a symplectic manifold; see \citet{OhLe2013} and \citet{Oh2015b}.

There are, however, two different ways of describing the dynamics.
In the formulation originally due to \citet{He1975a,He1976b}, elements in $\Sigma_{d}$ are written as is, i.e., one writes down ODEs for $\mathcal{A}$ and $\mathcal{B}$, whereas \citet{Ha1980, Ha1998} replaces $\mathcal{A} + {\rm i}\mathcal{B}$ by $P Q^{-1}$ with $d\times d$ complex matrices $Q$ and $P$ that satisfy certain relationships, and the corresponding dynamics is written in terms of $Q$ and $P$.
The geometry behind the two different parametrizations turns out to be precisely the fact that $\Sigma_{d}$ is the homogeneous space $\Sp(2d,\R)/\U(d)$, i.e., the set of variables $(Q,P)$ naturally lives in the symplectic group $\Sp(2d,\R)$ and $\mathcal{A} + {\rm i}\mathcal{B}$ is its projection to $\Sp(2d,\R)/\U(d)$; see Sections~\ref{sec:Sigma_d} and \ref{sec:parametrization_of_Sigma_d}, and also \citet{Oh2015b}.

As simple as the correspondence sounds, one encounters an obstacle when trying to interpret the two formulations from the symplectic/Hamiltonian point of view.
On one hand, it is fairly straightforward to formulate Heller's dynamics with $\mathcal{A}$ and $\mathcal{B}$ from the symplectic/Hamiltonian point of view because $\Sigma_{d}$ is a symplectic manifold; see \citet{OhLe2013}.
On the other hand, it is not so simple to do the same with the parameters $Q$ and $P$ of Hagedorn because the symplectic group $\Sp(2d,\R)$ is clearly not a symplectic manifold in general: Its dimension is $d(2d+1)$, which is odd when $d$ is odd.

\subsection{Main Results and Outline}
The main result of the paper is Theorem~\ref{thm:Siegel-MaWe}, which is stated at the beginning of Section~\ref{sec:Sigma_d-MaWe}: In short, we show that the Siegel upper half space $\Sigma_{d}$ is identified with the Marsden--Weinstein quotient arising from the cotangent bundle $T^{*}\R^{2d^{2}}$ with symmetry group $\Orth(2d)$.
Specifically, a certain level set of the momentum map and the corresponding isotropy group are identified with $\Sp(2d,\R)$ and $\U(d)$, respectively, thereby giving rise to the homogeneous space $\Sp(2d,\R)/\U(d)$ in the context of symplectic reduction.

The theorem also gives a clear symplectic/Hamiltonian picture of the connection between the two different parametrizations of the Gaussian wave packet dynamics by showing that one is a Hamiltonian system on $T^{*}\R^{d} \times T^{*}\R^{2d^{2}}$ with $\Orth(2d)$-symmetry and the other is its reduced Hamiltonian system on $T^{*}\R^{d} \times \Sigma_{d}$.

The outline of the paper is as follows:
We first review, in Section~\ref{sec:Sigma_d}, the geometry of the Siegel upper half space $\Sigma_{d}$ going through its realization as a homogeneous space; some of the results there will be later referred to when proving the main theorem.
Section~\ref{sec:Sigma_d-MaWe} states and proves the main result of this paper, Theorem~\ref{thm:Siegel-MaWe}.
Finally, in Section~\ref{sec:GWP}, we apply the theorem to the Gaussian wave packet dynamics to interpret the correspondence between two different formulations from the symplectic/Hamiltonian point of view.
We also exploit the Hamiltonian nature of the problem to derive a semiclassical angular momentum by applying Noether's theorem to semiclassical systems with rotational symmetry; this complements our earlier work on semiclassical angular momentum~\cite{Oh2015b}.

\section{The Siegel Upper Half Space $\Sigma_{d}$}
\label{sec:Sigma_d}
\subsection{Geometry of the Siegel Upper Half Space}
We first briefly review the well-known treatment of the Siegel upper half space $\Sigma_{d}$ as a homogeneous space (see \citet{Si1943} and also \citet[Section~4.5]{Fo1989} and \citet[Exercise~2.28 on p.~48]{McSa1999}).
Specifically, we show the following identification alluded above: 
\begin{equation*}
  \Sigma_{d} \cong \Sp(2d,\R)/\U(d),
\end{equation*}
where $\Sp(2d,\R)$ is the symplectic group of degree $2d$ over real numbers, i.e.,
\begin{equation*}
  \Sp(2d,\R) \defeq
  \setdef{
    S \in \Mat_{2d}(\R)
  }{ S^{T} J S = J}
  \quad\text{with}\quad
  J =
  \begin{bmatrix}
    0 & I_{d} \\
    -I_{d} & 0
  \end{bmatrix},
\end{equation*}
or equivalently, written as block matrices consisting of $d \times d$ submatrices,
\begin{equation}
  \label{def:Sp2d-block}
  \Sp(2d,\R) \defeq
  \setdef{
    \begin{bmatrix}
      A & B \\
      C & D
    \end{bmatrix}
    \in \Mat_{2d}(\R)
  }{ A^{T}C = C^{T}A,\, B^{T}D = D^{T}B,\, A^{T}D - C^{T}B = I_{d} },
\end{equation}
and $\U(d)$ is the unitary group of degree $d$.

Consider the (left) action of $\Sp(2d,\R)$ on $\Sigma_{d}$ defined by the generalized linear fractional transformation
\begin{equation*}
  \Psi\colon \Sp(2d,\R) \times \Sigma_{d} \to \Sigma_{d};
  \quad
  \parentheses{
    \begin{bmatrix}
      A & B \\
      C & D
    \end{bmatrix},
    \mathcal{X}
  }
  \mapsto
  (C + D\mathcal{X})(A + B\mathcal{X})^{-1}.
\end{equation*}
This action is transitive: By choosing
\begin{equation}
  \label{eq:X-mathcalAB}
  X \defeq
  \begin{bmatrix}
    A & B \\
    C & D
  \end{bmatrix}
  =
  \begin{bmatrix}
    I_{d} & 0 \\
    \mathcal{A} & I_{d}
  \end{bmatrix}
  \begin{bmatrix}
    \mathcal{B}^{-1/2} & 0 \\
    0 & \mathcal{B}^{1/2}
  \end{bmatrix}
  =
  \begin{bmatrix}
    \mathcal{B}^{-1/2} & 0 \\
    \mathcal{A}\mathcal{B}^{-1/2} & \mathcal{B}^{1/2}
  \end{bmatrix},
\end{equation}
which is easily shown to be symplectic, we have
\begin{equation*}
  \Psi_{X}({\rm i}I_{d}) = \mathcal{A} + {\rm i}\mathcal{B}.
\end{equation*}
The isotropy subgroup of the element ${\rm i}I_{d} \in \Sigma_{d}$ is given by
\begin{align}
  \Sp(2d,\R)_{{\rm i}I_{d}} &= \setdef{
    \begin{bmatrix}
      U  & V \\
      -V & U
    \end{bmatrix} \in \Mat_{2d}(\R)
  }{U^{T}U + V^{T}V = I_{d},\, U^{T}V = V^{T}U}
  \nonumber\\
  &= \Sp(2d,\R) \cap \mathsf{O}(2d),
    \label{eq:isotropy}
\end{align}
where $\mathsf{O}(2d)$ is the orthogonal group of degree $2d$; however $\Sp(2d,\R) \cap \mathsf{O}(2d)$ is identified with $\U(d)$ as follows:
\begin{equation}
  \label{eq:U(d)}
  \Sp(2d,\R) \cap \mathsf{O}(2d) \to \U(d);
  \quad
  \begin{bmatrix}
    U  & V \\
    -V & U
  \end{bmatrix}
  \mapsto U + {\rm i}V.
\end{equation}
Hence $\Sp(2d,\R)_{{\rm i}I_{d}} \cong \U(d)$ and thus $\Sigma_{d} \cong \Sp(2d,\R)/\U(d)$.
Indeed, we may identify $\Sp(2d,\R)/\U(d)$ with $\Sigma_{d}$ by the following map:
\begin{equation*}
  \Sp(2d,\R)/\U(d) \to \Sigma_{d};
  \quad
  [Y]_{\U(d)} \mapsto \Psi_{Y}({\rm i}I_{d}),
\end{equation*}
where $[\,\cdot\,]_{\U(d)}$ denotes a left coset of $\U(d)$ in $\Sp(2d,\R)$; then this gives rise to the explicit construction of the quotient map
\begin{equation}
  \label{eq:pi_Ud}
  \pi_{\U(d)}\colon \Sp(2d,\R) \to \Sp(2d,\R)/\U(d) \cong \Sigma_{d};
  \quad
  Y \mapsto \Psi_{Y}({\rm i}I_{d}),
\end{equation}
or more explicitly,
\begin{equation*}
  \pi_{\U(d)}\parentheses{
    \begin{bmatrix}
      A & B \\
      C & D
    \end{bmatrix}
  }
  = (C + {\rm i}D)(A + {\rm i}B)^{-1},
\end{equation*}
where $A + {\rm i}B$ can be shown to be invertible if $
\begin{tbmatrix}
  A & B \\
  C & D
\end{tbmatrix} \in \Sp(2d,\R)$.

As shown by \citet{Si1943}, $\Sigma_{d}$ is equipped with the Hermitian metric
\begin{equation*}
  g_{\Sigma_{d}} \defeq \tr\parentheses{ \mathcal{B}^{-1} \d\mathcal{C}\,\mathcal{B}^{-1} \d\bar{\mathcal{C}}\, }
  = \mathcal{B}^{-1}_{jl} \mathcal{B}^{-1}_{mk} \d\mathcal{C}_{lm} \otimes \d\bar{\mathcal{C}}_{jk}
\end{equation*}
and hence its imaginary part
\begin{equation}
  \label{eq:symplectic_form-Sigma_d}
  \Omega_{\Sigma_{d}} \defeq \Im g_{\Sigma_{d}}
  = \mathcal{B}^{-1}_{jl} \mathcal{B}^{-1}_{mk} \d\mathcal{B}_{lm} \wedge \d\mathcal{A}_{jk}
  = -\d\mathcal{B}^{-1}_{jk} \wedge \d\mathcal{A}_{jk}
\end{equation}
defines a symplectic form on $\Sigma_{d}$.

\subsection{The Iwasawa Decomposition and the Siegel Upper Half Space}
We may make the above geometric structure more explicit by making use of the so-called {\em Iwasawa decomposition} of symplectic matrices~(see, e.g., \citet[Section~2.2.2]{Go2006}).
The Iwasawa decomposition renders any symplectic matrix $X_{1} \defeq \left[\begin{smallmatrix}
  A & B \\
  C & D
\end{smallmatrix}\right]
\in \Sp(2d,\R)$ into the following factorization of symplectic matrices:
\begin{align*}
  X_{1} = 
  \begin{bmatrix}
    A & B \\
    C & D
  \end{bmatrix}
  &=
  \begin{bmatrix}
    I_{d} & 0 \\
    P & I_{d}
  \end{bmatrix}
  \begin{bmatrix}
    L & 0 \\
    0 & L^{-1}
  \end{bmatrix}
  \begin{bmatrix}
    U & V \\
    -V & U
  \end{bmatrix}
  \\
  &=
  \begin{bmatrix}
    L & 0 \\
    PL & L^{-1}
  \end{bmatrix}
  \begin{bmatrix}
    U & V \\
    -V & U
  \end{bmatrix}
  = X_{2}\, \mathcal{U},
\end{align*}
where $P$ and $L$ are symmetric $d \times d$ matrices given by\footnote{The matrix $A A^{T} + B B^{T}$ is always invertible if $X_{1}$ is in $\Sp(2d,\R)$ and hence so is $L$; see, e.g., \citet[Section~2.1.2]{Go2006}.}
\begin{gather*}
  P = (C A^{T} + D B^{T})(A A^{T} + B B^{T})^{-1},
  \qquad
  L = (A A^{T} + B B^{T})^{1/2},
\end{gather*}
and $X_{2}$ and $\mathcal{U} \in \U(d)$ are defined by\footnote{Note that, in the above expression $X_{2}\,\mathcal{U}$, the matrix $\mathcal{U}$ is seen as an element in $\Sp(2d,\R) \cap \Orth(2d)$ via \eqref{eq:U(d)}.}
\begin{equation*}
  X_{2} \defeq
  \begin{bmatrix}
    L & 0 \\
    PL & L^{-1}
  \end{bmatrix},
  \qquad
  \mathcal{U} \defeq U + {\rm i}V = (A A^{T} + B B^{T})^{-1/2}(A + {\rm i}B).
\end{equation*}
Then we see that the matrices $X_{1}$ and $X_{2}$ define the same coset in $\Sp(2d, \R)/\U(d)$, i.e., $[X_{1}]_{\U(d)} = [X_{2}]_{\U(d)}$ because $\Psi_{X_{1}}({\rm i}I_{d}) = \Psi_{X_{2}} \circ \Psi_{\mathcal{U}}({\rm i}I_{d}) = \Psi_{X_{2}}({\rm i}I_{d})$.

Now, take any element $\mathcal{A} + {\rm i}\mathcal{B}$ in the Siegel upper half space $\Sigma_{d}$ and suppose that $\pi_{\U(d)}(X_{1}) = \mathcal{A} + {\rm i}\mathcal{B}$.
Then $\pi_{\U(d)}(X_{1}) = \pi_{\U(d)}(X_{2}) = P + {\rm i}L^{-2}$ using \eqref{eq:X-mathcalAB} and so $P = \mathcal{A}$ and $L = \mathcal{B}^{-1/2}$.
This observation leads to the following expression for those elements in $\Sp(2d,\R)$ that project to $\mathcal{A} + {\rm i}\mathcal{B} \in \Sigma_{d}$ via $\pi_{\U(d)}$:
\begin{equation}
  \label{eq:Iwasawa}
  \pi_{\U(d)}^{-1}(\mathcal{A} + {\rm i}\mathcal{B}) = \setdef{
    \begin{bmatrix}
      \mathcal{B}^{-1/2} & 0 \\
      \mathcal{A}\mathcal{B}^{-1/2} & \mathcal{B}^{1/2}
    \end{bmatrix}
    \begin{bmatrix}
      U  & V \\
      -V & U
    \end{bmatrix}
    \in \Sp(2d,\R)
  }{
    U + {\rm i}V \in \U(d)
  }.
\end{equation}
This expression will be later useful in coordinate calculations.

\section{The Siegel Upper Half Space $\Sigma_{d}$ and Symplectic Reduction}
\label{sec:Sigma_d-MaWe}
In this section, we prove our main result, that the Siegel upper half space $\Sigma_{d}$ is identified with a reduced symplectic manifold in the Marsden--Weinstein~\cite{MaWe1974} sense.
More specifically:
\begin{theorem}
  \label{thm:Siegel-MaWe}
  Let $\mathcal{Z}_{d} \defeq T^{*}\R^{2d^{2}}$ be the cotangent bundle of the vector space $\R^{2d^{2}}$ with the standard symplectic form $\Omega_{\mathcal{Z}_{d}}$, where we identify $\R^{2d^{2}}$ with the set $\Mat_{d \times 2d}(\R)$ of $d \times 2d$ real matrices and $\mathcal{Z}_{d}$ with the set $\Mat_{2d}(\R)$ of $2d \times 2d$ real matrices.
  Consider the action of the orthogonal group $\Orth(2d)$ on $\mathcal{Z}_{d}$ defined by matrix multiplication from the right, and let ${\bf M}\colon \mathcal{Z}_{d} \to \orth(2d)^{*}$ be  the corresponding momentum map.
  Then the Marsden--Weinstein quotient $\overline{\mathcal{Z}}_{d}^{J} \defeq {\bf M}^{-1}(J)/\Orth(2d)_{J}$ at $J = \begin{tbmatrix}
    0 & I_{d} \\
    -I_{d} & 0
  \end{tbmatrix} \in \orth(2d)^{*}$ is the Siegel upper half space $\Sigma_{d}$ with symplectic form $\overline{\Omega}_{J} = -\frac{1}{2}\Omega_{\Sigma_{d}}$, i.e.,
  \begin{equation*}
    \pi_{J}^{*} \overline{\Omega}_{J} = i_{J}^{*} \Omega_{\mathcal{Z}_{d}},
  \end{equation*}
  where $i_{J}\colon {\bf M}^{-1}(J) \hookrightarrow \mathcal{Z}_{d}$ is the inclusion and $\pi_{J}\colon {\bf M}^{-1}(J) \to \overline{\mathcal{Z}}_{d}^{J}$ is the projection to the quotient.
\end{theorem}
We prove the above theorem in the rest of the section.

\subsection{Basic Setup}
Consider the real vector space $\R^{2d^{2}}$; we prefer to write each element in $\R^{2d^{2}}$ as a $d \times 2d$ block matrix consisting of two $d \times d$ submatrices, i.e.,
\begin{equation*}
  \R^{2d^{2}}
  \cong \Mat_{d \times 2d}(\R)
  \defeq \setdef{
    \begin{bmatrix}
      Q_{1} & Q_{2}
    \end{bmatrix}
  }{
    Q_{1}, Q_{2} \in \Mat_{d}(\R)
  }.
\end{equation*}
Let $\mathcal{Z}_{d} \defeq T^{*}\R^{2d^{2}}$ be the cotangent bundle of $\R^{2d^{2}}$.
Then each element $Z$ in $\mathcal{Z}_{d}$ is identified with a $2d \times 2d$ block matrix consisting of four $d \times d$ submatrices, i.e.,
\begin{equation}
  \label{eq:mathcalZ}
  \mathcal{Z}_{d} \defeq T^{*}\R^{2d^{2}} \cong \Mat_{2d}(\R) = \setdef{
    Z \defeq 
    \begin{bmatrix}
      Q_{1} & Q_{2} \smallskip\\
      P_{1} & P_{2}
    \end{bmatrix}
  }{
    Q_{1}, Q_{2}, P_{1}, P_{2} \in \Mat_{d}(\R)
  }.
\end{equation}
The standard symplectic structure on $\mathcal{Z}_{d}$ is given by
\begin{equation}
  \label{eq:Omega-Z_d}
  \Omega_{\mathcal{Z}_{d}}
  \defeq \d{Q_{1}} \wedge \d{P_{1}}  + \d{Q_{2}} \wedge \d{P_{2}}
  = \d{Q_{1}^{jk}} \wedge \d{P_{1,jk}}  + \d{Q_{2}^{jk}} \wedge \d{P_{2,jk}},
\end{equation}
where $j$ and $k$ run from $1$ to $d$, and $Q_{l}^{jk}$ and $P_{l,jk}$ stand for the $(j,k)$-entries of the matrices $Q_{l}$ and $P_{l}$, respectively, for $l = 1, 2$.
With the canonical one-form $\Theta_{\mathcal{Z}_{d}}$ on $\mathcal{Z}_{d}$ defined by
\begin{equation*}
  \Theta_{\mathcal{Z}_{d}} \defeq \tr(P_{1}^{T}\d{Q_{1}}) + \tr(P_{2}^{T}\d{Q_{2}})
  = P_{1,jk}\d{Q_{1}^{jk}} + P_{2,jk}\d{Q_{2}^{jk}},
\end{equation*}
one can write the symplectic form $\Omega_{\mathcal{Z}_{d}}$ as
\begin{equation*}
  \Omega_{\mathcal{Z}_{d}} = -\d\Theta_{\mathcal{Z}_{d}}. 
\end{equation*}

\subsection{$\Orth(2d)$-action and Momentum Map}
Consider the action of the orthogonal group $\Orth(2d)$ on $\mathcal{Z}_{d}$ defined by matrix multiplication from the right, i.e.,
\begin{equation}
  \label{eq:Phi}
  \Phi\colon \Orth(2d) \times \mathcal{Z}_{d} \to \mathcal{Z}_{d};
  \qquad
  (\mathcal{R}, Z) \mapsto Z \mathcal{R}.
\end{equation}
It is a straightforward calculation to see that $\Phi$ leaves the canonical one-form $\Theta_{\mathcal{Z}_{d}}$ invariant, i.e., $\Phi_{\mathcal{R}}^{*}\Theta_{\mathcal{Z}_{d}} = \Theta_{\mathcal{Z}_{d}}$ for any $\mathcal{R} \in \Orth(2d)$, and hence is a symplectic action with respect to the symplectic form $\Omega_{\mathcal{Z}_{d}}$ defined in \eqref{eq:Omega-Z_d}, i.e., $\Phi_{\mathcal{R}}^{*}\Omega_{\mathcal{Z}_{d}} = \Omega_{\mathcal{Z}_{d}}$ for any $\mathcal{R} \in \Orth(2d)$.

What is the momentum map corresponding to the $\Orth(2d)$-action $\Phi$?
Let $\orth(2d)$ be the Lie algebra of $\Orth(2d)$ and $\orth(2d)^{*}$ be the dual of $\orth(2d)$.
For any $\xi \in \orth(2d)$, the corresponding infinitesimal generator $\xi_{\mathcal{Z}_{d}}$ is given by
\begin{equation*}
  \xi_{\mathcal{Z}_{d}}(Z) \defeq \left. \od{}{\eps} \Phi_{\exp(\eps\xi)}(Z) \right|_{\eps=0}
  = Z\xi,
\end{equation*}
where $Z\xi$ stands for the standard matrix multiplication.
Then, since $\mathcal{Z}_{d}$ is an exact symplectic manifold with the symplectic structure $\Omega_{\mathcal{Z}_{d}} = -\d\Theta_{\mathcal{Z}_{d}}$ and the action $\Phi$ leaves $\Theta_{\mathcal{Z}_{d}}$ invariant, the corresponding momentum map ${\bf M} \colon \mathcal{Z}_{d} \to \orth(2d)^{*}$ may be defined as follows (see, e.g., \citet[Theorem~4.2.10 on p.~282]{AbMa1978}): For any $\xi \in \orth(2d)$,
\begin{equation*}
  \ip{ {\bf M}(Z) }{ \xi } = \Theta_{\mathcal{Z}_{d}}\parentheses{ \xi_{\mathcal{Z}_{d}}(Z) }
  \quad\text{or}\quad
  \ip{ {\bf M}(\cdot) }{ \xi } = \ins{\xi_{\mathcal{Z}_{d}}}\Theta_{\mathcal{Z}_{d}}.
\end{equation*}
We equip the Lie algebra $\orth(2d)$ with the inner product
\begin{equation}
  \label{eq:inner_product-o2d}
  \ip{\cdot}{\cdot}\colon \orth(2d) \times \orth(2d) \to \R;
  \qquad
  (\xi, \eta) \mapsto \ip{\xi}{\eta} \defeq \frac{1}{2}\tr(\xi^{T}\eta).
\end{equation}
Then we may identify the dual $\orth(2d)^{*}$ with $\orth(2d)$ itself via the inner product.
So we may write the components of the momentum map ${\bf M}$ as follows:
\begin{equation*}
  {\bf M}\colon \mathcal{Z}_{d} \to \orth(2d)^{*};
  \qquad
  Z = 
  \begin{bmatrix}
    Q_{1} & Q_{2} \smallskip\\
    P_{1} & P_{2}
  \end{bmatrix}
  \mapsto
  \begin{bmatrix}
    M_{11} & M_{12} \smallskip\\
    -M_{12}^{T} & M_{22}
  \end{bmatrix},
\end{equation*}
where both $M_{11}$ and $M_{22}$ are skew-symmetric $d \times d$ matrices and $M_{12}$ is a $d \times d$ matrix, i.e., $M_{11}, M_{22} \in \orth(d)$ and $M_{12} \in \Mat_{d}(\R)$.

Let us first find $M_{11}$.
Pick $\xi = \begin{tbmatrix}
  \xi_{11} & 0 \\
  0 & 0
\end{tbmatrix} \in \orth(2d)$ with $\xi_{11} \in \orth(d)$; then
\begin{equation*}
  \ip{ {\bf M}(Z) }{ \xi }
  = \frac{1}{2}\tr\parentheses{
    \begin{bmatrix}
      M_{11}^{T}\,\xi_{11} & 0 \smallskip\\
      M_{12}^{T}\,\xi_{11} & 0
    \end{bmatrix}
  }
  = \frac{1}{2}\tr\parentheses{ M_{11}^{T}\,\xi_{11} }.
\end{equation*}
On the other hand, the corresponding infinitesimal generator is given by
\begin{equation*}
  \xi_{\mathcal{Z}_{d}}(Z) =
  \begin{bmatrix}
    Q_{1} & Q_{2} \smallskip\\
    P_{1} & P_{2}
  \end{bmatrix}
  \begin{bmatrix}
    \xi_{11} & 0 \smallskip\\
    0 & 0
  \end{bmatrix}
  =
  \begin{bmatrix}
    Q_{1}\,\xi_{11} & 0 \smallskip\\
    P_{1}\,\xi_{11} & 0
  \end{bmatrix}
\end{equation*}
and hence 
\begin{equation*}
  \Theta_{\mathcal{Z}_{d}}\parentheses{ \xi_{\mathcal{Z}_{d}}(Z) }
  = \tr(P_{1}^{T}Q_{1}\xi_{11})
  = \frac{1}{2}\tr\brackets{ (Q_{1}^{T}P_{1} - P_{1}^{T}Q_{1})^{T} \xi_{11} }.
\end{equation*}
Since $\xi_{11} \in \orth(d)$ is arbitrary, we find $M_{11} = Q_{1}^{T}P_{1} - P_{1}^{T}Q_{1}$.

Likewise, $\xi = \begin{tbmatrix}
  0 & 0 \\
  0 & \xi_{22}
\end{tbmatrix} \in \orth(2d)$ with $\xi_{22} \in \orth(d)$ yields $M_{22} = Q_{2}^{T}P_{2} - P_{2}^{T}Q_{2}$.

Finally, taking $\xi = \begin{tbmatrix}
  0 & \xi_{12} \\
  -\xi_{12}^T & 0
\end{tbmatrix} \in \orth(2d)$ with $\xi_{12} \in \Mat_{d}(\R)$, we have
\begin{equation*}
  \ip{ {\bf M}(Z) }{ \xi }
  = \frac{1}{2}\tr\parentheses{
    \begin{bmatrix}
      M_{12}\,\xi_{12}^{T} & M_{11}^{T}\,\xi_{12} \smallskip\\
      -M_{22}^{T}\,\xi_{12}^{T} & M_{12}^{T}\,\xi_{12}
    \end{bmatrix}
  }
  = \tr\parentheses{ M_{12}^{T}\,\xi_{12} },
\end{equation*}
whereas the corresponding infinitesimal generator is
\begin{equation*}
  \xi_{\mathcal{Z}_{d}}(Z) =
  \begin{bmatrix}
    Q_{1} & Q_{2} \smallskip\\
    P_{1} & P_{2}
  \end{bmatrix}
  \begin{bmatrix}
    0 & \xi_{12} \smallskip\\
    -\xi_{12}^T & 0
  \end{bmatrix}
  =
  \begin{bmatrix}
    -Q_{2}\,\xi_{12}^{T} & Q_{1}\,\xi_{12} \smallskip\\
    -P_{2}\,\xi_{12}^{T} & P_{1}\,\xi_{12}
  \end{bmatrix}
\end{equation*}
and so
\begin{equation*}
  \Theta_{\mathcal{Z}_{d}}\parentheses{ \xi_{\mathcal{Z}_{d}}(Z) }
  = \tr(-P_{1}^{T}Q_{2}\,\xi_{12}^{T}) + \tr(P_{2}^{T}Q_{1}\,\xi_{12})
  = \tr\brackets{ (Q_{1}^{T}P_{2} - P_{1}^{T}Q_{2})^{T} \xi_{12} }.
\end{equation*}
Again, since $\xi_{12} \in \Mat_{d}(\R)$ is arbitrary, we find $M_{12} = Q_{1}^{T}P_{2} - P_{1}^{T}Q_{2}$.

As a result, we have the momentum map
\begin{equation}
  \label{eq:M}
  {\bf M}\colon \mathcal{Z}_{d} \to \orth(2d)^{*};
  \qquad
  Z = 
  \begin{bmatrix}
    Q_{1} & Q_{2} \smallskip\\
    P_{1} & P_{2}
  \end{bmatrix}
  \mapsto
  \begin{bmatrix}
    Q_{1}^{T}P_{1} - P_{1}^{T}Q_{1} & Q_{1}^{T}P_{2} - P_{1}^{T}Q_{2} \smallskip\\
    -(Q_{1}^{T}P_{2} - P_{1}^{T}Q_{2})^{T} & Q_{2}^{T}P_{2} - P_{2}^{T}Q_{2}
  \end{bmatrix}.
\end{equation}
It is a straightforward calculation to check that ${\bf M}$ is equivariant, i.e.,
\begin{equation*}
  {\bf M} \circ \Phi_{\mathcal{R}} = \Ad_{\mathcal{R}}^{*} \circ {\bf M}.
\end{equation*}

\subsection{Momentum Level Set and Reduced Space}
\label{ssec:MomentumLevelSetAndReducedSpace}
Now let us look at the level set of the momentum map ${\bf M}$ at $J = \begin{tbmatrix}
  0 & I_{d} \\
  -I_{d} & 0
\end{tbmatrix} \in \orth(2d)^{*}$, i.e.,
\begin{equation*}
  {\bf M}^{-1}(J)
  = \setdef{
    \begin{bmatrix}
      Q_{1} & Q_{2} \smallskip\\
      P_{1} & P_{2}
    \end{bmatrix}
    \in \Mat_{2d}(\R)
  }{
    Q_{1}^{T}P_{1} = P_{1}^{T}Q_{1},\,
    Q_{2}^{T}P_{2} = P_{2}^{T}Q_{2},\,
    Q_{1}^{T}P_{2} - P_{1}^{T}Q_{2} = I_{d}
  }.
\end{equation*}
This is precisely the definition of the symplectic group $\Sp(2d,\R)$ in terms of block matrices shown in \eqref{def:Sp2d-block}, i.e., ${\bf M}^{-1}(J) = \Sp(2d,\R)$.

Moreover, the coadjoint isotropy subgroup $\Orth(2d)_{J}$ of $J \in \orth(2d)^{*}$ is easily identified as
\begin{align*}
  \Orth(2d)_{J} &= \setdef{
    \mathcal{R} \in \Orth(2d)
  }{
    \Ad_{\mathcal{R}}^{*}J = \mathcal{R}^{T} J \mathcal{R} = J
  }
  \\
  &= \Sp(2d,\R) \cap \Orth(2d) \cong \U(d),
\end{align*}
which is precisely the isotropy subgroup $\Sp(2d,\R)_{{\rm i}I_{d}}$ of the action of $\Sp(2d,\R)$ on the Siegel upper half space $\Sigma_{d}$ shown in \eqref{eq:isotropy}.
The action of the coadjoint isotropy subgroup $\Orth(2d)_{J}$ on the momentum level set ${\bf M}^{-1}(J) = \Sp(2d,\R)$ is free, as $S \mathcal{R} = S$ for $S \in \Sp(2d,\R)$ and $\mathcal{R} \in \Orth(2d)$ implies $\mathcal{R} = I_{2d}$; the action is also proper as well because $\U(d)$ is compact.

So we may now invoke the Marsden--Weinstein reduction~\cite{MaWe1974} (see also \citet{Me1973} and \citet[Sections~1.1 \& 1.2]{MaMiOrPeRa2007}) to conclude that the reduced space or the Marsden--Weinstein quotient $\overline{\mathcal{Z}}_{d}^{J} \defeq {\bf M}^{-1}(J)/\Orth(2d)_{J}$ is a symplectic manifold, but then this quotient coincides with the Siegel upper half space $\Sigma_{d} \cong \Sp(2d,\R)/\U(d)$, i.e.,
\begin{equation*}
  \overline{\mathcal{Z}}_{d}^{J} \defeq {\bf M}^{-1}(J)/\Orth(2d)_{J} = \Sp(2d,\R)/\U(d) \cong \Sigma_{d}.
\end{equation*}

\subsection{Reduced Symplectic Form}
\label{ssec:ReducedSymplecticForm}
Let us define the following inclusion and projection maps:
\begin{equation*}
  i_{J}\colon {\bf M}^{-1}(J) \hookrightarrow \mathcal{Z}_{d},
  \qquad
  \pi_{J}\colon {\bf M}^{-1}(J) \to {\bf M}^{-1}(J)/\Orth(2d)_{J} \eqdef \overline{\mathcal{Z}}_{d}^{J}.
\end{equation*}
As shown in \citet{MaWe1974}, the symplectic form $\overline{\Omega}_{J}$ on the reduced symplectic manifold $\overline{\mathcal{Z}}_{d}^{J}$ corresponding to the original one $\Omega_{\mathcal{Z}_{d}}$ is uniquely characterized as follows:
\begin{equation}
  \label{eq:Omega_J}
  \pi_{J}^{*} \overline{\Omega}_{J} = i_{J}^{*} \Omega_{\mathcal{Z}_{d}}.
\end{equation}
We would like to find an expression for $\overline{\Omega}_{J}$.
Let $\mathcal{A} + {\rm i}\mathcal{B}$ be an arbitrary element in $\overline{\mathcal{Z}}_{d}^{J} \cong \Sigma_{d}$.
Since $\pi_{\U(d)}$ in \eqref{eq:pi_Ud} and $\pi_{J}$ are identical, those elements in $\Sp(2d,\R)$ that project to $\mathcal{A} + {\rm i}\mathcal{B}$ are written as in \eqref{eq:Iwasawa}.
But then this implies that any element in $\Sp(2d,\R) = {\bf M}^{-1}(J) = \pi_{J}^{-1}(\Sigma_{d})$ is written as
\begin{equation*}
  \begin{bmatrix}
   \mathcal{B}^{-1/2} & 0 \smallskip\\
   \mathcal{A}\mathcal{B}^{-1/2} & \mathcal{B}^{1/2}
  \end{bmatrix}
  \begin{bmatrix}
    U  & V \smallskip\\
    -V & U
  \end{bmatrix}
  = 
  \begin{bmatrix}
    \mathcal{B}^{-1/2} U & \mathcal{B}^{-1/2} V \smallskip\\
    \mathcal{A}\mathcal{B}^{-1/2} U - \mathcal{B}^{1/2} V & \mathcal{A}\mathcal{B}^{-1/2} V + \mathcal{B}^{1/2} U
  \end{bmatrix},
\end{equation*}
with some $\mathcal{A} + {\rm i}\mathcal{B} \in \Sigma_{d}$ and $U + {\rm i}V \in \U(d)$, i.e., $U^{T}U + V^{T}V = I_{d}$ and $U^{T}V = V^{T}U$.
In other words, the above expression gives an expression for the inclusion $i_{J}\colon {\bf M}^{-1}(J) \hookrightarrow \mathcal{Z}_{d}$ in terms of the coordinates adapted to the horizontal and vertical directions of the principal bundle $\pi_{J}\colon {\bf M}^{-1}(J) \to {\bf M}^{-1}(J)/\U(d)$.
Then the pull-back by $i_{J}$ of the one-form $\Theta_{\mathcal{Z}_{d}}$ is written as
\begin{equation*}
  i_{J}^{*}\Theta_{\mathcal{Z}_{d}} = \frac{1}{2}\tr( \mathcal{A}\, \d\mathcal{B}^{-1} )  + \tr( U^{T}\d{V} - V^{T}\d{U} ),
\end{equation*}
and so, taking into account the relationships between $U$ and $V$,
\begin{equation*}
  i_{J}^{*}\Omega_{\mathcal{Z}_{d}} = -\d(i_{J}^{*}\Theta_{\mathcal{Z}_{d}}) = \frac{1}{2} \d\mathcal{B}^{-1}_{jk} \wedge \d\mathcal{A}_{jk},
\end{equation*}
and hence we have, from \eqref{eq:Omega_J} and \eqref{eq:symplectic_form-Sigma_d}, $\pi_{J}^{*}\overline{\Omega}_{J} = -\frac{1}{2}\pi_{J}^{*} \Omega_{\Sigma_{d}}$.
Since $\pi_{J}$ is a surjective submersion, $\pi_{J}^{*}$ is injective; thus we obtain $\overline{\Omega}_{J} = -\frac{1}{2} \Omega_{\Sigma_{d}}$.
This completes the proof of Theorem~\ref{thm:Siegel-MaWe}.

\section{Application to Gaussian Wave Packet Dynamics}
\label{sec:GWP}
\subsection{The Gaussian Wave Packet Dynamics}
Consider the time-dependent Schr\"odinger equation
\begin{equation}
  \label{eq:Schroedinger}
  {\rm i}\hbar\,\pd{}{t}\psi(x,t) = -\frac{\hbar^{2}}{2m} \Delta \psi(x,t) + V(x)\,\psi(x,t)
\end{equation}
for the wave function $\psi(x,t)$ under the potential $V(x)$, where $\hbar > 0$ is Planck's constant, $t \ge 0$ is the time, $x \in \R^{d}$ is the position in the physical space $\R^{d}$, and $\Delta$ stands for the Laplacian in $\R^{d}$.
Our motivation for identifying the Siegel upper half space $\Sigma_{d}$ as a Marsden--Weinstein quotient comes from a geometric description of the dynamics of the Gaussian wave packet ansatz
\begin{equation}
  \label{eq:psi_0}
  \psi_{0} \defeq \parentheses{ \frac{\det\mathcal{B}}{(\pi\hbar)^{d}} }^{1/4} \exp\braces{ \frac{{\rm i}}{\hbar}\brackets{ \frac{1}{2}(x - q)^{T}(\mathcal{A} + {\rm i}\mathcal{B})(x - q) + p \cdot (x - q) + \phi } }
\end{equation}
for \eqref{eq:Schroedinger}; the factor in front of the exponential normalizes the wave function, i.e., $\norm{ \psi_{0} } = 1$ as an element in $L^{2}(\R^{d})$, and $\psi_{0}$ is parametrized by $(q,p) \in T^{*}\R^{d}$, $\mathcal{C} \defeq \mathcal{A} + {\rm i}\mathcal{B} \in \Sigma_{d}$, and $\phi \in \mathbb{S}^{1}$.
It is well known (see \citet{He1975a,He1976b}) that, when $V$ is quadratic, the Gaussian wave packet~\eqref{eq:psi_0} is an {\em exact} solution of the Schr\"odinger equation~\eqref{eq:Schroedinger} if the parameters $(q, p, \mathcal{A}, \mathcal{B})$ satisfy the set of ODEs
\begin{equation}
  \label{eq:Heller}
  \begin{array}{c}
    \DS
    \dot{q} = \frac{p}{m},
    \qquad
    \dot{p} = -\nabla V(q),
    \medskip\\
    \DS
    \dot{\mathcal{A}} = -\frac{1}{m}(\mathcal{A}^{2} - \mathcal{B}^{2}) - \nabla^{2}V(q),
    \qquad
    \dot{\mathcal{B}} = -\frac{1}{m}(\mathcal{A}\mathcal{B} + \mathcal{B}\mathcal{A}),
  \end{array}
\end{equation}
where $\nabla^{2}V$ is the Hessian matrix of $V$, and the phase $\phi(t)$ is determined by
\begin{equation*}
  \phi(t) = \phi(0) + \int_{0}^{t} \brackets{ \frac{p(s)^{2}}{2m} - V(q(s)) - \frac{\hbar}{2m} \tr(\mathcal{B}(s)) }\,ds.
\end{equation*}

\citet{Ha1980, Ha1998}, on the other hand, has a slightly different parametrization of the Gaussian wave packet~\eqref{eq:psi_0}:
\begin{equation}
  \label{eq:psi_0-Hagedorn}
  \psi_{0} = (\pi\hbar)^{-d/4} (\det Q)^{-1/2} \exp\braces{ \frac{{\rm i}}{\hbar}\brackets{ \frac{1}{2}(x - q)^{T}P Q^{-1}(x - q) + p \cdot (x - q) + S } },
\end{equation}
where $Q$ and $P$ are complex $d \times d$ matrices, i.e., $Q, P \in \Mat_{d}(\C)$, that satisfy
\begin{equation}
  \label{eq:Q_P-Sp}
  Q^{T}P - P^{T}Q = 0
  \quad\text{and}\quad
  Q^{*}P - P^{*}Q = 2{\rm i}I_{d},
\end{equation}
and an appropriate branch cut is taken for $(\det Q)^{1/2}$; also the new parameter $S$ is defined as
\begin{equation*}
  S \defeq \phi - \frac{\hbar}{2}\arg(\det Q).
\end{equation*}
\citet{Ha1980, Ha1998} showed that \eqref{eq:psi_0-Hagedorn} is an exact solution of the Schr\"odinger equation if the potential $V$ is quadratic and also the parameters $(q, p, Q, P)$ satisfy
\begin{equation}
  \label{eq:Hagedorn}
  \dot{q} = \frac{p}{m},
  \qquad
  \dot{p} = -\nabla V(q),
  \qquad
  \dot{Q} = \frac{P}{m},
  \qquad
  \dot{P} =  -\nabla^{2}V(q)\,Q,
\end{equation}
and the quantity $S(t)$ is the classical action integral evaluated along the solution $(q(t),p(t))$, i.e.,
\begin{equation*}
  S(t) = S(0) + \int_{0}^{t} \parentheses{ \frac{p(s)^{2}}{2m} - V(q(s)) }\,ds.
\end{equation*}
It is also shown by \citet{Ha1980, Ha1998} that \eqref{eq:psi_0-Hagedorn} with \eqref{eq:Hagedorn} gives an $O(t\sqrt{\hbar})$ approximation when the potential $V$ is not quadratic as long as it satisfies some regularity assumptions.

\subsection{Parametrization of the Siegel Upper Half Space}
\label{sec:parametrization_of_Sigma_d}
The replacement of $\mathcal{A} + {\rm i}\mathcal{B} \in \Sigma_{d}$ in \eqref{eq:psi_0} by $P Q^{-1}$ in \eqref{eq:psi_0-Hagedorn} has a simple geometric interpretation.
\citet[Section~V.1]{Lu2008} (see also \citet{Oh2015b}) pointed out that the conditions~\eqref{eq:Q_P-Sp} for the matrices $Q$ and $P$ are precisely the conditions for the matrix $\begin{tbmatrix}
  \Re Q & \Im Q \smallskip\\
  \Re P & \Im P
\end{tbmatrix}$ to be symplectic, i.e., 
\begin{align*}
  \Sp(2d,\R)
  &= \setdef{
    \begin{bmatrix}
      \Re Q & \Im Q \smallskip\\
      \Re P & \Im P
    \end{bmatrix}
  }{Q, P \in \Mat_{d}(\C),\, Q^{T}P - P^{T}Q = 0,\, Q^{*}P - P^{*}Q = 2{\rm i}I_{d}}.
\end{align*}
In fact, the projection of these elements to $\Sigma_{d}$ by $\pi_{\U(d)}\colon \Sp(2d,\R) \to \Sigma_{d}$ in \eqref{eq:pi_Ud} gives
\begin{equation}
  \label{eq:pi_Ud-2}
  \pi_{\U(d)}\parentheses{
    \begin{bmatrix}
      \Re Q & \Im Q \smallskip\\
      \Re P & \Im P
    \end{bmatrix}
  }
  = P Q^{-1}.
\end{equation}
It is also easy to show that the dynamics of $\mathcal{A}$ and $\mathcal{B}$ defined in \eqref{eq:Heller} is the projection to $\Sigma_{d}$ by $\pi_{\U(d)}$ of the dynamics of $Q$ and $P$ in \eqref{eq:Hagedorn}; conversely, the dynamics \eqref{eq:Hagedorn} is a proper lift to $\Sp(2d,\R)$ of the dynamics \eqref{eq:Heller} in some appropriate sense; see \citet{Oh2015b}.

\subsection{Interpretation as a Hamiltonian Reduction}
Theorem~\ref{thm:Siegel-MaWe} sheds a new perspective on the above connection between the equations~\eqref{eq:Heller} and \eqref{eq:Hagedorn} of Heller and Hagedorn, respectively, in terms of Hamiltonian reduction.

Let us first introduce a new parametrization of the space $\mathcal{Z}_{d}$:
The observation~\eqref{eq:pi_Ud-2} from the previous subsection motivates us to rewrite the definition~\eqref{eq:mathcalZ} of the space $\mathcal{Z}_{d}$ as follows:
\begin{equation*}
  \mathcal{Z}_{d} \defeq T^{*}\R^{2d^{2}}
  = \setdef{
    Z = 
    \begin{bmatrix}
      \Re Q & \Im Q \smallskip\\
      \Re P & \Im P
    \end{bmatrix}
  }{
    Q, P \in \Mat_{d}(\C)
  },
\end{equation*}
i.e., $Q_{1} + {\rm i}Q_{2} = Q$ and $P_{1} + {\rm i}P_{2} = P$.

Here we assume that the potential $V$ is quadratic for simplicity, and will consider the general case in the next subsection.
If $V$ is quadratic then the Hessian $\nabla^{2}V$ is a constant matrix; as a result, the system~\eqref{eq:Heller} decouples into the classical Hamiltonian system in $T^{*}\R^{d} = \{ (q, p) \}$ and
\begin{equation}
  \label{eq:Heller-A_B}
  \dot{\mathcal{A}} = -\frac{1}{m}(\mathcal{A}^{2} - \mathcal{B}^{2}) - \nabla^{2}V(q),
  \qquad
  \dot{\mathcal{B}} = -\frac{1}{m}(\mathcal{A}\mathcal{B} + \mathcal{B}\mathcal{A})
\end{equation}
in $\Sigma_{d} = \{ (\mathcal{A}, \mathcal{B}) \}$, and similarly, \eqref{eq:Hagedorn} decouples into the classical Hamiltonian system in $T^{*}\R^{d} = \{ (q, p) \}$ and
\begin{equation}
  \label{eq:Hagedorn-Q_P}
  \dot{Q} = \frac{P}{m},
  \qquad
  \dot{P} =  -\nabla^{2}V\,Q
\end{equation}
in $\mathcal{Z}_{d} = \{ (Q, P) \}$.
It turns out that the dynamics~\eqref{eq:Heller-A_B} in $\Sigma_{d}$ is obtained by Hamiltonian reduction of the dynamics~\eqref{eq:Hagedorn-Q_P} in $\mathcal{Z}_{d}$:
\begin{proposition}[Reduction of Gaussian wave packet dynamics---quadratic potentials]
  \label{prop:quadratic}
  Suppose that the potential $V$ is quadratic, and let $H_{\mathcal{Z}_{d}}\colon \mathcal{Z}_{d} \to \R$ be the Hamiltonian defined by
  \begin{equation}
    \label{eq:H_Z_d}
    H_{\mathcal{Z}_{d}}(Z) \defeq \frac{1}{2m}\tr(P^{*}P) + \frac{1}{2}\tr(Q^{*}\, \nabla^{2}V\, Q).
  \end{equation}
  Then:
  \begin{enumerate}
    \renewcommand{\theenumi}{\roman{enumi}}
  \item The Hamiltonian vector field $X_{H_{\mathcal{Z}_{d}}}$ on $\mathcal{Z}_{d}$ defined by the Hamiltonian system
    \begin{equation}
      \label{eq:HamSys_mathcalZ}
      \ins{X_{H_{\mathcal{Z}_{d}}}}\Omega_{\mathcal{Z}_{d}} = \d{H_{\mathcal{Z}_{d}}}
    \end{equation}
     gives the equations~\eqref{eq:Hagedorn-Q_P} of Hagedorn.
    \smallskip
  \item The Hamiltonian $H_{\mathcal{Z}_{d}}$ is invariant under the $\Orth(2d)$-action $\Phi$ defined in \eqref{eq:Phi}, i.e., $H_{\mathcal{Z}_{d}} \circ \Phi_{\mathcal{R}} = H_{\mathcal{Z}_{d}}$ for any $\mathcal{R} \in \Orth(2d)$, and hence the Hamiltonian system~\eqref{eq:HamSys_mathcalZ} conserves the corresponding momentum map ${\bf M}$ in \eqref{eq:M}; particularly, ${\bf M}^{-1}(J) = \Sp(2d,\R)$ is an invariant manifold of the system~\eqref{eq:HamSys_mathcalZ}.
    \smallskip
  \item The reduced Hamiltonian $\overline{H}_{J}\colon \overline{\mathcal{Z}}_{d}^{J} \to \R$, i.e., the function $\overline{H}_{J}$ uniquely characterized by
    \begin{equation}
      \label{def:H_J}
      \overline{H}_{J} \circ \pi_{J} = H_{\mathcal{Z}_{d}} \circ i_{J},
    \end{equation}
    takes the form
    \begin{equation}
      \label{eq:H_J}
      \overline{H}_{J}(\mathcal{A}, \mathcal{B}) = \frac{1}{2} \tr\brackets{ \mathcal{B}^{-1}\parentheses{ \frac{\mathcal{A}^{2} + \mathcal{B}^{2}}{m} + \nabla^{2}V } }.
    \end{equation}
    \smallskip
  \item The Hamiltonian vector field $X_{\overline{H}_{J}}$ on the reduced space $\overline{\mathcal{Z}}_{d}^{J} \cong \Sigma_{d}$ defined by the Hamiltonian system
    \begin{equation}
      \label{eq:RedHamSys}
      \ins{X_{\overline{H}_{J}}}\overline{\Omega}_{J} = \d\overline{H}_{J}
    \end{equation}
    gives the equations~\eqref{eq:Heller-A_B} of Heller.
  \end{enumerate}
\end{proposition}

\begin{proof}
  The first assertion follows from straightforward calculations; so is the $\Orth(2d)$-invariance of the Hamiltonian~\eqref{eq:H_Z_d}.

  The remaining assertions follow easily from Theorem~\ref{thm:Siegel-MaWe}:
  By Noether's Theorem (see, e.g., \citet[Theorem~11.4.1 on p.~372]{MaRa1999}), the $\Orth(2d)$-invariance of the Hamiltonian $H_{\mathcal{Z}_{d}}$ implies that the momentum map ${\bf M}$ is conserved along the flow defined by $X_{H_{\mathcal{Z}_{d}}}$; we note that this result is observed by \citet[Lemma~V.1.4 on p.~126]{Lu2008} via direct calculations.
  We have already shown in Section~\ref{ssec:MomentumLevelSetAndReducedSpace} that ${\bf M}^{-1}(J) = \Sp(2d,\R)$.

  The reduced Hamiltonian~\eqref{eq:H_J} follows from the defining relation~\eqref{def:H_J} and coordinate calculations that are similar to those performed in Section~\ref{ssec:ReducedSymplecticForm} when finding the reduced symplectic form $\overline{\Omega}_{J}$.

  Finally, that the reduced Hamiltonian system is given by \eqref{eq:RedHamSys} is standard in symplectic reduction~\cite{MaWe1974}, and straightforward calculations yield the last assertion.
\end{proof}

\begin{remark}
  For the Hamiltonian dynamics~\eqref{eq:HamSys_mathcalZ} to be interpreted as the dynamics of the Gaussian wave packet~\eqref{eq:psi_0-Hagedorn}, one needs to restrict the initial condition to the invariant manifold ${\bf M}^{-1}(J) = \Sp(2d,\R) \subset \mathcal{Z}_{d}$, which is equivalent to \eqref{eq:Q_P-Sp} as mentioned above.
  This guarantees that $Q$ is invertible and also that $P Q^{-1} \in \Sigma_{d}$, and hence $\psi_{0} \in L^{2}(\R^{d})$.
\end{remark}

\begin{remark}
  Setting $\mathcal{C} \defeq \mathcal{A} + {\rm i}\mathcal{B}$, \eqref{eq:Heller-A_B} is written as a matrix Riccati equation, i.e.,
  \begin{equation*}
    \dot{\mathcal{C}} = -\frac{1}{m}\mathcal{C}^{2} - \nabla^{2}V.
  \end{equation*}
  Its Hamiltonian lift \eqref{eq:Hagedorn-Q_P} to $\Sp(2d,\R) \subset \mathcal{Z}_{d}$ is linear in $Q$ and $P$, and may be considered as an example of the Hirota bilinearization of the matrix Riccati equation; see, e.g., \citet{Hi1979, Hi2000, Hi2004}.
\end{remark}

\subsection{Full Dynamics as Hamiltonian Systems}
What if the potential $V$ is {\em not} quadratic?
It turns out that, with a slight modification in the equations \eqref{eq:Heller} and \eqref{eq:Hagedorn}, these systems may also be rendered Hamiltonian as well, and again one is the reduced version of the other.

Let $\Omega_{T^{*}\R^{d}} \defeq \d{q^{i}} \wedge \d{p_{i}}$ be the standard symplectic form on $T^{*}\R^{d}$, $\pr_{1}\colon T^{*}\R^{d} \times \mathcal{Z}_{d} \to T^{*}\R^{d}$ and $\pr_{2}\colon T^{*}\R^{d} \times \mathcal{Z}_{d} \to \mathcal{Z}_{d}$ be the projections, and define a symplectic form $\Omega$ on $T^{*}\R^{d} \times \mathcal{Z}_{d}$ by
\begin{equation}
  \label{eq:Omega}
  \Omega \defeq \pr_{1}^{*} \Omega_{T^{*}\R^{d}} + \frac{\hbar}{2}\pr_{2}^{*} \Omega_{\mathcal{Z}_{d}}
\end{equation}
and a Hamiltonian $H\colon T^{*}\R^{d} \times \mathcal{Z}_{d} \to \R$ by
\begin{align}
  H &= \frac{p^{2}}{2m} + V(q) + \frac{\hbar}{2}H_{\mathcal{Z}_{d}}
  \nonumber\\
  &= \frac{p^{2}}{2m} + \frac{\hbar}{4m}\tr(P^{*}P)
    + V(q) + \frac{\hbar}{4} \tr(Q^{*}\, \nabla^{2}V\, Q).
    \label{eq:H}
\end{align}
Note that the symplectic form $\Omega$ is written as $\Omega = -\d\Theta$ with
\begin{equation}
  \label{eq:Theta}
  \Theta \defeq \pr_{1}^{*} \Theta_{T^{*}\R^{d}} + \frac{\hbar}{2}\pr_{2}^{*} \Theta_{\mathcal{Z}_{d}}
  = p_{j}\,\d{q^{j}} + \frac{\hbar}{2}\brackets{ 
    \tr(P_{1}^{T}\d{Q_{1}}) + \tr(P_{2}^{T}\d{Q_{2}})
  }.
\end{equation}

\begin{proposition}[Reduction of Gaussian wave packet dynamics---general potentials]
  Let $V \in C^{3}(\R^{d})$ and $H\colon T^{*}\R^{d} \times \mathcal{Z}_{d} \to \R$ be the Hamiltonian defined in \eqref{eq:H}.
  \begin{enumerate}
    \renewcommand{\theenumi}{\roman{enumi}}
  \item The Hamiltonian system
    \begin{subequations}
      \label{eq:Hagedorn2}
      \begin{equation}
        \ins{X_{H}}\Omega = \d{H}
      \end{equation}
      gives the system
      \begin{equation}
        \begin{array}{c}
          \DS
          \dot{q} = \frac{p}{m},
          \qquad
          \dot{p} = -\pd{}{q} \braces{ V(q) + \frac{\hbar}{4}\tr[Q^{*} \nabla^{2}V(q) Q] },
          \medskip\\
          \DS
          \dot{Q} = \frac{P}{m},
          \qquad
          \dot{P} =  -\nabla^{2}V(q)\,Q.
        \end{array}
      \end{equation}
    \end{subequations}
    \smallskip
  \item The Hamiltonian~\eqref{eq:H} is $\Orth(2d)$-invariant under the $\Orth(2d)$-action $\id_{T^{*}\R^{d}} \times \Phi$ on $T^{*}\R^{d} \times \mathcal{Z}_{d}$ and hence the corresponding momentum map $\tilde{\bf M}\colon T^{*}\R^{d} \times \mathcal{Z}_{d} \to \orth(2d)^{*}$, which is given by $\tilde{\bf M} = {\bf M} \circ \pr_{2}$, is conserved along the flow of the system~\eqref{eq:Hagedorn2}.
    In particular, $\tilde{\bf M}^{-1}(J) = T^{*}\R^{d} \times \Sp(2d,\R)$ is an invariant manifold of \eqref{eq:Hagedorn2}.
    \smallskip
  \item Symplectic reduction by the $\Orth(2d)$-symmetry at the value $J \in \orth(2d)^{*}$ yields the reduced symplectic manifold $T^{*}\R^{d} \times \Sigma_{d}$ with symplectic form
    \begin{equation}
      \label{eq:overline_Omega}
      \overline{\Omega} \defeq \d{q^{i}} \wedge \d{p_{i}} + \frac{\hbar}{4}\, \d\mathcal{B}^{-1}_{jk} \wedge \d\mathcal{A}_{jk}
    \end{equation}
    and the reduced Hamiltonian
    \begin{equation}
      \label{eq:overline_H}
      \overline{H} = \frac{p^{2}}{2m} + V(q)
      + \frac{\hbar}{4}\tr\brackets{ \mathcal{B}^{-1}\parentheses{ \frac{\mathcal{A}^{2} + \mathcal{B}^{2}}{m} + \nabla^{2}V(q) } },
    \end{equation}
    i.e., they are uniquely determined by $\mathcal{I}_{J}^{*} \Omega = \Pi_{J}^{*} \overline{\Omega}$ and $H \circ \mathcal{I}_{J} = \overline{H} \circ \Pi_{J}$, where $\mathcal{I}_{J}\colon \tilde{\bf M}^{-1}(J) \hookrightarrow T^{*}\R^{d} \times \mathcal{Z}_{d}$ and $\Pi_{J}\colon \tilde{\bf M}^{-1}(J) \to T^{*}\R^{d} \times \Sigma_{d}$ are the inclusion and projection, respectively, defined by
    \begin{equation*}
      \mathcal{I}_{J} \defeq \id_{T^{*}\R^{d}} \times i_{J},
      \qquad
      \Pi_{J} \defeq \id_{T^{*}\R^{d}} \times \pi_{J}.
    \end{equation*}
    \smallskip
  \item The reduced Hamiltonian system
    \begin{subequations}
      \label{eq:Heller2}
      \begin{equation}
        \ins{X_{\overline{H}}}\overline{\Omega} = \d{\overline{H}}
      \end{equation}
      gives
      \begin{equation}
        \begin{array}{c}
          \DS
          \dot{q} = \frac{p}{m},
          \qquad
          \dot{p} = -\pd{}{q}\brackets{ V(q) + \frac{\hbar}{4} \tr\parentheses{ \mathcal{B}^{-1} \nabla^{2}V(q) } },
          \medskip\\
          \DS
          \dot{\mathcal{A}} = -\frac{1}{m}(\mathcal{A}^{2} - \mathcal{B}^{2}) - \nabla^{2}V(q),
          \qquad
          \dot{\mathcal{B}} = -\frac{1}{m}(\mathcal{A}\mathcal{B} + \mathcal{B}\mathcal{A}).
        \end{array}
      \end{equation}
    \end{subequations}
    \smallskip
  \item Particularly, when the potential $V$ is quadratic, the systems \eqref{eq:Hagedorn2} and \eqref{eq:Heller2} recover \eqref{eq:Hagedorn} and \eqref{eq:Heller}, respectively.
  \end{enumerate}
\end{proposition}

\begin{proof}
  Most of the assertions are straightforward generalizations of those in Proposition~\ref{prop:quadratic} and so are proved in a similar way.
  The last assertion follows because the Hessian $\nabla^{2}V$ becomes constant when the potential $V$ is quadratic.
\end{proof}

\begin{remark}
  The above symplectic form $\overline{\Omega}$ and Hamiltonian $\overline{H}$ are exactly those that appeared in \citet{OhLe2013} (see $\overline{\Omega}_{\hbar}$ in (17) of Theorem~4.1 and also $\overline{H}_{1}$ in Section~7.1 there; see also \citet[Section~II.C]{Oh2015b}), but are derived from a different point of view.
  It was shown in \citet{OhLe2013} that the symplectic form $\overline{\Omega}$ in \eqref{eq:overline_Omega} is a pull-back of the symplectic form on the projective Hilbert space $\mathbb{P}(L^{2}(\R^{d}))$ induced by the symplectic form
  \begin{equation*}
    \Omega(\psi_{1}, \psi_{2}) = 2\hbar \Im\ip{\psi_{1}}{\psi_{2}}
  \end{equation*}
  on $L^{2}(\R^{d})$---the Schr\"odinger equation~\eqref{eq:Schroedinger} is written as a Hamiltonian system on $L^{2}(\R^{d})$ or $\mathbb{P}(L^{2}(\R^{d}))$ in terms of those symplectic forms; see, e.g., \citet[Chapter~2]{MaRa1999}.
  As for the Hamiltonian $\overline{H}$ in \eqref{eq:overline_H}, with a certain decay property assumed for the potential $V$, one can show that $\overline{H}$ is an $O(\hbar^{2})$ approximation to the expectation value of the Hamiltonian operator
  \begin{equation*}
    \hat{H} \defeq -\frac{\hbar^{2}}{2m} \Delta + V(x)
  \end{equation*}
  with respect to the Gaussian~\eqref{eq:psi_0}, i.e.,
  \begin{equation*}
    \ip{\psi_{0}}{\hat{H}\psi_{0}} = \overline{H} + O(\hbar^{2}),
  \end{equation*}
  where $\ip{\cdot}{\cdot}$ stands for the inner product on $L^{2}(\R^{d})$; see \citet[Proposition~7.1]{OhLe2013}.
\end{remark}

\subsection{Rotational Symmetry and Conservation Law}
The $\Orth(2d)$ symmetry exploited above is rather an intrinsic symmetry of the Hamiltonian system \eqref{eq:Hagedorn2}.
Here we would like to see what happens if the system has a symmetry in its physical configuration, particularly when the potential $V$ has a rotational symmetry.
This leads to the semiclassical angular momentum for the semiclassical system~\eqref{eq:Hagedorn2}, which corresponds to the one for the reduced system \eqref{eq:Heller2} found in \citet{Oh2015b}:
\begin{proposition}[Rotational symmetry and semiclassical angular momentum]
  Let $\varphi\colon \SO(d) \times \R^{d} \to \R^{d}$ be the natural action of the rotation group $\SO(d)$ on the configuration space $\R^{d}$, i.e., for any $R \in \SO(d)$,
  \begin{equation*}
    \varphi_{R}\colon \R^{d} \to \R^{d};
    \quad
    q \mapsto R q,
  \end{equation*}
  and suppose that the potential $V \in C^{3}(\R^{d})$ is invariant under the $\SO(d)$-action, i.e.,
  \begin{equation*}
    V \circ \varphi_{R} = V.
  \end{equation*}
  Also let $\Upsilon\colon \SO(d) \times (T^{*}\R^{d} \times \mathcal{Z}_{d}) \to T^{*}\R^{d} \times \mathcal{Z}_{d}$ be the $\SO(d)$-action on $T^{*}\R^{d} \times \mathcal{Z}_{d}$ defined as follows: For any $R \in \SO(d)$,
  \begin{equation*}
    \Upsilon_{R}\colon T^{*}\R^{d} \times \mathcal{Z}_{d} \to T^{*}\R^{d} \times \mathcal{Z}_{d};
    \qquad
    (q, p, Q, P) \mapsto (R q, R p, R Q, R P).
  \end{equation*}
  Then:
  \begin{enumerate}
    \renewcommand{\theenumi}{\roman{enumi}}
  \item $\Upsilon$ leaves the canonical one-form $\Theta$ invariant, and hence is a symplectic action with respect to the symplectic form~\eqref{eq:Omega}, i.e., $\Upsilon_{R}^{*} \Omega = \Omega$ for any $R \in \SO(d)$.
    \smallskip
  \item The Hamiltonian~\eqref{eq:H} is invariant under the action, i.e., $H \circ \Upsilon_{R} = H$ for any $R \in \SO(d)$.
    \smallskip
  \item The semiclassical system~\eqref{eq:Hagedorn2} conserves the {\em semiclassical angular momentum} ${\bf J}\colon T^{*}\R^{d} \times \mathcal{Z}_{d} \to \so(d)^{*}$ defined by
    \begin{equation}
      \label{eq:J-Q_P}
      {\bf J}(q, p, Q, P) = q \diamond p + \frac{\hbar}{2}(P_{1}Q_{1}^{T} + P_{2}Q_{2}^{T} - Q_{1}P_{1}^{T} - Q_{2}P_{2}^{T}),
    \end{equation}
    where $q \diamond p$ denotes the $d \times d$ matrix defined by $(q \diamond p)_{ij} \defeq q_{j}p_{i} - q_{i}p_{j}$; see, e.g., \citet[Remark~6.3.3 on p.~150]{Ho2011b}.
    The angular momentum map ${\bf J}$ is equivariant as well, i.e., for any $R \in \SO(d)$,
    \begin{equation*}
      {\bf J} \circ \Upsilon_{R} = \Ad_{R^{-1}}^{*} {\bf J}.
    \end{equation*}
  \end{enumerate}
\end{proposition}

\begin{proof}
  It is straightforward computations to see, for any $R \in \SO(d)$, that $\Upsilon_{R}^{*} \Theta = \Theta$ and so $\Upsilon_{R}^{*} \Omega = \Omega$ as well as that $H \circ \Upsilon_{R} = H$.
  For the symplecticity of $\Upsilon$, one may alternatively identify $T^{*}\R^{d} \times \mathcal{Z}_{d}$ with the cotangent bundle $T^{*}(\R^{d} \times \Mat_{d \times 2d}(\R))$ and then see that $\Upsilon$ is the cotangent lift of the action
  \begin{equation*}
    \SO(d) \times (\R^{d} \times \Mat_{d \times 2d}) \to (\R^{d} \times \Mat_{d \times 2d});
    \qquad
    (R, (q, Q_{1}, Q_{2})) \mapsto (R q, R Q_{1}, R Q_{2});
  \end{equation*}
  see, e.g., \citet[Proposition~6.3.2 on p.~170]{MaRa1999}.

  Let us find the corresponding momentum map, i.e., the angular momentum for the semiclassical system~\eqref{eq:Hagedorn2}.
  For any $\xi \in \so(d)$, its infinitesimal generator on $T^{*}\R^{d} \times \mathcal{Z}_{d}$ is given by
  \begin{align*}
    \xi_{T^{*}\R^{d} \times \mathcal{Z}_{d}}(q, p, Q, P)
    &\defeq \left. \od{}{\eps} \Upsilon_{\exp(\eps\xi)}(q, p, Q, P) \right|_{\eps=0}
    \\
    &= (\xi q)^{j} \pd{}{q^{j}}
    + (\xi p)_{j} \pd{}{p_{j}}
      + \sum_{l=1}^{2} \parentheses{
      (\xi Q_{l})^{jk} \pd{}{Q_{l}^{jk}}
      + (\xi P_{l})_{jk} \pd{}{P_{l,jk}}
    },
  \end{align*}
  where the indices $j$ and $k$ run from $1$ to $d$.
  Since $T^{*}\R^{d} \times \mathcal{Z}_{d}$ is an exact symplectic manifold with symplectic form $\Omega = -\d\Theta$ with $\Theta$ given in \eqref{eq:Theta}, and $\Upsilon$ leaves $\Theta$ invariant, the corresponding angular momentum ${\bf J}\colon T^{*}\R^{d} \times \mathcal{Z}_{d} \to \so(d)^{*}$ is given by
  \begin{equation*}
    \ip{ {\bf J}(q, p, Q, P) }{ \xi } = \Theta\parentheses{ \xi_{T^{*}\R^{d} \times \mathcal{Z}_{d}}(q, p, Q, P) }
    \quad\text{or}\quad
    \ip{ {\bf J}(\cdot) }{ \xi } = \ins{\xi_{T^{*}\R^{d} \times \mathcal{Z}_{d}}}\Theta,
  \end{equation*}
  and straightforward computations yield
  \begin{align*}
    \ip{ {\bf J}(q, p, Q, P) }{ \xi }
    &= p\cdot \xi q + \frac{\hbar}{2}\tr( P_{1}^{T} \xi Q_{1} + P_{2}^{T} \xi Q_{2} )
    \\
    &= \ip{ q \diamond p + \frac{\hbar}{2}(P_{1}Q_{1}^{T} + P_{2}Q_{2}^{T} - Q_{1}P_{1}^{T} - Q_{2}P_{2}^{T}) }{ \xi },
  \end{align*}
  where we identified $\so(d)^{*}$ with $\so(d)$ via the inner product as in \eqref{eq:inner_product-o2d}:
  \begin{equation*}
    \ip{\cdot}{\cdot}\colon \so(d) \times \so(d) \to \R;
    \qquad
    (\xi, \eta) \mapsto \ip{\xi}{\eta} \defeq \frac{1}{2}\tr(\xi^{T}\eta).
  \end{equation*}
  So we obtain \eqref{eq:J-Q_P}; it is a conserved quantity of the system~\eqref{eq:Hagedorn2} due to the $\SO(d)$-invariance of the Hamiltonian $H$ and Noether's Theorem (see, e.g., \citet[Theorem~11.4.1 on p.~372]{MaRa1999}).

  The equivariance of ${\bf J}$ is easy to show by direct calculations:
  \begin{equation*}
    {\bf J}(R q, R p, R Q, R P) = R\,{\bf J}(q, p, Q, P) R^{T}.
  \end{equation*}
  The equivariance also follows from the fact that $\Upsilon$ is a cotangent lift as mentioned earlier; see, e.g., \citet[Theorem~12.1.4 on p.~386]{MaRa1999}.
\end{proof}

Assuming that $(q, p, Q, P)$ is in the invariant manifold $\tilde{\bf M}^{-1}(J) = T^{*}\R^{d} \times \Sp(2d,\R)$, i.e., $\begin{tbmatrix}
  \Re Q & \Im Q \smallskip\\
  \Re P & \Im P
\end{tbmatrix} \in \Sp(2d,\R)$ or equivalently \eqref{eq:Q_P-Sp}, the setup and result in the above proposition descend from $\tilde{\bf M}^{-1}(J)$ to $T^{*}\R^{d} \times \Sigma_{d}$ and recover the semiclassical angular momentum found in \citet{Oh2015b}.
In fact, the action $\Upsilon$ induces an $\SO(d)$-action $\Gamma\colon \SO(d) \times (T^{*}\R^{d} \times \Sigma_{d}) \to T^{*}\R^{d} \times \Sigma_{d}$ so that the diagram
\begin{equation*}
  \begin{tikzcd}[column sep=10ex, row sep=7ex]
    T^{*}\R^{d} \times \Sp(2d,\R) \arrow{r}{} \arrow{d}[swap]{\Pi_{J}} \arrow{r}{\Upsilon_{R}|_{\tilde{\bf M}^{-1}(J)}} & T^{*}\R^{d} \times \Sp(2d,\R) \arrow{d}{\Pi_{J}}
    \\
    T^{*}\R^{d} \times \Sigma_{d} \arrow{r}[swap]{\Gamma_{R}} & T^{*}\R^{d} \times \Sigma_{d}
  \end{tikzcd}
\end{equation*}
commutes: For any $R \in \SO(d)$, we have
\begin{equation}
  \label{eq:Gamma}
  \Gamma_{R}\colon T^{*}\R^{d} \times \Sigma_{d} \to T^{*}\R^{d} \times \Sigma_{d};
  \qquad
  (q, p, \mathcal{A}, \mathcal{B}) \mapsto (R q, R p, R\mathcal{A}R^{T}, R\mathcal{B}R^{T}).
\end{equation}
This coincides with (26) in \citet{Oh2015b}; in fact $\Gamma$ is a symplectic action with respect to the symplectic form~\eqref{eq:overline_Omega}, i.e., $\Gamma_{R}^{*} \overline{\Omega} = \overline{\Omega}$ for any $R \in \SO(d)$.

Notice that
\begin{equation*}
  \Re(P Q^{*} - Q P^{*}) = P_{1}Q_{1}^{T} + P_{2}Q_{2}^{T} - Q_{1}P_{1}^{T} - Q_{2}P_{2}^{T}
\end{equation*}
with $Q = Q_{1} + {\rm i}Q_{2}$ and $P = P_{1} + {\rm i}P_{2}$ as defined above.
Then \eqref{eq:Q_P-Sp} and \eqref{eq:pi_Ud-2} give (recall that $\pi_{\U(d)} = \pi_{J}$)
\begin{equation*}
  \pi_{J}(Q, P) = P Q^{-1} = \mathcal{A} + {\rm i}\mathcal{B}
\end{equation*}
as well as $Q Q^{*} = \mathcal{B}^{-1}$ because (\citet[Lemma~V.1.1 on p.~124]{Lu2008})
\begin{align*}
  \mathcal{B} Q Q^{*} &= \Im( P Q^{-1} ) Q Q^{*} \\
                      &= \frac{1}{2{\rm i}}( P Q^{-1} - (Q^{*})^{-1} P^{*} ) Q Q^{*} \\
                      &= \frac{1}{2{\rm i}}( P Q^{*} - (Q^{*})^{-1} P^{*} Q Q^{*} ) \\
                      &= \frac{1}{2{\rm i}}[ P Q^{*} - (Q^{*})^{-1} (Q^{*} P - 2{\rm i}I_{d}) Q^{*} ]
                        = I_{d},
\end{align*}
where we used the second equality of \eqref{eq:Q_P-Sp}.
Therefore,
\begin{align*}
  P Q^{*} - Q P^{*}
  &= P Q^{-1} (Q Q^{*}) - Q (P^{*} Q) Q^{-1}
  \\
    &= P Q^{-1} (Q Q^{*}) - Q (Q^{*}P - 2{\rm i}I_{d}) Q^{-1}
  \\
  &= ( \mathcal{A} + {\rm i}\mathcal{B}) \mathcal{B}^{-1} - \mathcal{B}^{-1}( \mathcal{A} + {\rm i}\mathcal{B}) + 2{\rm i}I_{d}
  \\
  &= [\mathcal{A}, \mathcal{B}^{-1}] + 2{\rm i}I_{d},
\end{align*}
and hence we have ${\bf J}|_{T^{*}\R^{d} \times \Sp(2d,\R)} = \overline{\bf J} \circ \Pi_{J}$ with $\overline{\bf J}\colon T^{*}\R^{d} \times \Sigma_{d} \to \so(d)^{*}$ defined by
\begin{equation*}
  \overline{\bf J}(q, p, \mathcal{A}, \mathcal{B}) = q \diamond p - \frac{\hbar}{2}[\mathcal{B}^{-1}, \mathcal{A}].
\end{equation*}
This is exactly the semiclassical angular momentum for \eqref{eq:Heller2} derived via the action \eqref{eq:Gamma} in \citet[Theorem~3.1]{Oh2015b}.

The above semiclassical angular momentum is a natural one in the quantum mechanical sense as well, as mentioned in \citet[Section~III.B]{Oh2015b}.
In fact, set $d = 3$ and let $\hat{x}$ and $\hat{p}\defeq -{\rm i} \hbar \nabla$ be the position and angular operators; then one can show that, again assuming $(q, p, Q, P) \in \tilde{\bf M}^{-1}(J)$,
\begin{equation*}
  {\bf J}(q, p, Q, P) = \overline{\bf J}(q, p, \mathcal{A}, \mathcal{B}) = \ip{\psi_{0}}{(\hat{x} \times \hat{p}) \psi_{0}},
\end{equation*}
where $\psi_{0}$ is either \eqref{eq:psi_0} or \eqref{eq:psi_0-Hagedorn}.

\section*{Acknowledgments}
I would like to thank Melvin Leok for helpful discussions.
This work was partially supported by the AMS--Simons Travel Grant.

\bibliography{Siegel}

\begin{thebibliography}{20}
\providecommand{\natexlab}[1]{#1}
\providecommand{\url}[1]{\texttt{#1}}
\expandafter\ifx\csname urlstyle\endcsname\relax
  \providecommand{\doi}[1]{doi: #1}\else
  \providecommand{\doi}{doi: \begingroup \urlstyle{rm}\Url}\fi

\bibitem[Abraham and Marsden(1978)]{AbMa1978}
R.~Abraham and J.~E. Marsden.
\newblock \emph{Foundations of Mechanics}.
\newblock Addison--Wesley, 2nd edition, 1978.

\bibitem[de~Gosson(2006)]{Go2006}
M.~A. de~Gosson.
\newblock \emph{Symplectic Geometry and Quantum Mechanics}.
\newblock Birkh{\"a}user, 2006.

\bibitem[Folland(1989)]{Fo1989}
G.~B. Folland.
\newblock \emph{Harmonic Analysis in Phase Space}.
\newblock Princeton University Press, 1989.

\bibitem[Hagedorn(1980)]{Ha1980}
G.~A. Hagedorn.
\newblock Semiclassical quantum mechanics.
\newblock \emph{Communications in Mathematical Physics}, 71\penalty0
  (1):\penalty0 77--93, 1980.

\bibitem[Hagedorn(1998)]{Ha1998}
G.~A. Hagedorn.
\newblock Raising and lowering operators for semiclassical wave packets.
\newblock \emph{Annals of Physics}, 269\penalty0 (1):\penalty0 77--104, 1998.

\bibitem[Heller(1975)]{He1975a}
E.~J. Heller.
\newblock Time-dependent approach to semiclassical dynamics.
\newblock \emph{Journal of Chemical Physics}, 62\penalty0 (4):\penalty0
  1544--1555, 1975.

\bibitem[Heller(1976)]{He1976b}
E.~J. Heller.
\newblock Classical {$S$}-matrix limit of wave packet dynamics.
\newblock \emph{Journal of Chemical Physics}, 65\penalty0 (11):\penalty0
  4979--4989, 1976.

\bibitem[Hirota(1979)]{Hi1979}
R.~Hirota.
\newblock Nonlinear partial difference equations. v. nonlinear equations
  reducible to linear equations.
\newblock \emph{Journal of the Physical Society of Japan}, 46\penalty0
  (1):\penalty0 312--319, 1979.

\bibitem[Hirota(2000)]{Hi2000}
R.~Hirota.
\newblock \emph{Lectures on Finite Difference Equations (in Japanese)}.
\newblock Saiensu-sha Publishers, 2000.

\bibitem[Hirota(2004)]{Hi2004}
R.~Hirota.
\newblock \emph{The Direct Method in Soliton Theory}.
\newblock Cambridge University Press, 2004.

\bibitem[Holm(2011)]{Ho2011b}
D.~D. Holm.
\newblock \emph{Geometric Mechanics, Part II: Rotating, Translating and
  Rolling}.
\newblock Imperial College Press, 2nd edition edition, 2011.

\bibitem[Lubich(2008)]{Lu2008}
C.~Lubich.
\newblock \emph{From quantum to classical molecular dynamics: reduced models
  and numerical analysis}.
\newblock European Mathematical Society, Z{\"u}rich, Switzerland, 2008.

\bibitem[Marsden and Ratiu(1999)]{MaRa1999}
J.~E. Marsden and T.~S. Ratiu.
\newblock \emph{Introduction to Mechanics and Symmetry}.
\newblock Springer, 1999.

\bibitem[Marsden and Weinstein(1974)]{MaWe1974}
J.~E. Marsden and A.~Weinstein.
\newblock Reduction of symplectic manifolds with symmetry.
\newblock \emph{Reports on Mathematical Physics}, 5\penalty0 (1):\penalty0
  121--130, 1974.

\bibitem[Marsden et~al.(2007)Marsden, Misiolek, Ortega, Perlmutter, and
  Ratiu]{MaMiOrPeRa2007}
J.~E. Marsden, G.~Misiolek, J.~P. Ortega, M.~Perlmutter, and T.~S. Ratiu.
\newblock \emph{Hamiltonian Reduction by Stages}.
\newblock Springer, 2007.

\bibitem[McDuff and Salamon(1999)]{McSa1999}
D.~McDuff and D.~Salamon.
\newblock \emph{Introduction to Symplectic Topology}.
\newblock Oxford University Press, 1999.

\bibitem[Meyer(1973)]{Me1973}
K.~R. Meyer.
\newblock Symmetries and integrals in mechanics.
\newblock In M.~Peixoto, editor, \emph{Dynamical Systems}. Academic Press,
  1973.

\bibitem[Ohsawa(2015)]{Oh2015b}
T.~Ohsawa.
\newblock Symmetry and conservation laws in semiclassical wave packet dynamics.
\newblock \emph{Journal of Mathematical Physics}, 56\penalty0 (3):\penalty0
  032103, 2015.

\bibitem[Ohsawa and Leok(2013)]{OhLe2013}
T.~Ohsawa and M.~Leok.
\newblock Symplectic semiclassical wave packet dynamics.
\newblock \emph{Journal of Physics A: Mathematical and Theoretical},
  46\penalty0 (40):\penalty0 405201, 2013.

\bibitem[Siegel(1943)]{Si1943}
C.~L. Siegel.
\newblock Symplectic geometry.
\newblock \emph{American Journal of Mathematics}, 65\penalty0 (1):\penalty0
  1--86, 1943.

\end{thebibliography}
\bibliographystyle{plainnat}

\end{document}